\documentclass[%
 reprint,
superscriptaddress,
nofootinbib,
 amsmath,amssymb,
 aps,
 pra,
floatfix,
]{revtex4-2}

\def\label#1{%
 \@bsphack
  \protected@write\@auxout{}{%
   \string\newlabel{#1}{{\@currentlabel}{\thepage}{}{}{}}%
  }%
 \@esphack
}%

\usepackage{amsthm}
\usepackage{tikz}
\usetikzlibrary{decorations.pathmorphing}
\tikzset{snake it/.style={decorate, decoration=snake}}
\newtheorem{theorem}{Theorem}

\newtheorem{lemma}{Lemma}

\newtheorem{definition}{Definition}
\usepackage{braket}
\usepackage{bbm}
\usepackage[
	 pagebackref=false,
	 colorlinks=true,
 linkcolor=blue,
 urlcolor=blue,
 filecolor=black,
 citecolor=blue,
 pdfstartview=FitV,
 pdfauthor={},
 pdfsubject={},
 pdfkeywords={},
 pdfpagemode=None,
 bookmarksopen=true, breaklinks
 ]{hyperref}

\usepackage[normalem]{ulem}
\usepackage{amsmath}
\usepackage{enumerate}
\usepackage{amsfonts}
\usepackage{epsfig}
\usepackage{mathbbol}
\usepackage{braket}
\usepackage{tikz}

\newcommand{\be}{\begin{equation}}
\newcommand{\ee}{\end{equation}}
\newcommand{\bea}{\begin{eqnarray}}
\newcommand{\eea}{\end{eqnarray}}
\newcommand{\bega}{\begin{gather}}
\newcommand{\eega}{\end{gather}}
\newcommand{\nn}{\nonumber}
\newcommand{\bi}{\begin{itemize}}
\newcommand{\ei}{\end{itemize}}
\newcommand{\ben}{\begin{enumerate}}
\newcommand{\een}{\end{enumerate}}
\newcommand{\bca}{\begin{cases}}
\newcommand{\eca}{\end{cases}}
\newcommand{\bln}{\begin{align}}
\newcommand{\eln}{\end{align}}
\newcommand{\bst}{\begin{split}}
\newcommand{\est}{\end{split}}
\def\ie{\begin{equation}\begin{aligned}}
\def\fe{\end{aligned}\end{equation}}
\newcommand{\bma}{\le(\begin{matrix}}
\newcommand{\ema}{\end{matrix}\ri)}

\def\le{\left}
\def\ri{\right}

\newcommand{\beq}{\begin{equation}}
\newcommand{\eeq}{\end{equation}}

\newcommand{\norm}[1]{\parallel #1 \parallel}

\usepackage{anyfontsize}

\begin{document}

\newcommand\saa[1]{{\textcolor{teal}{[\textbf{SAA}: #1]}}}
\newcommand\MK[1]{{\color{red}[\textbf{MK}: #1]}}
\newcommand\SL[1]{{\textcolor{olive}{[\textbf{SL}: #1]}}}

\newtheorem{conjecture}{Conjecture}
\newtheorem{example}{Example}

 \newcommand{\nyu}[1]{
	\centerline{
		\begin{minipage}[c]{0.7\textwidth}
			\begin{center}
			${}^{#1}$ Center for Cosmology and Particle Physics, New York University, New York, NY 10003, USA
			\end{center}
		\end{minipage}
		}}
\usetikzlibrary{decorations.pathmorphing}
\usetikzlibrary{calc}

\title{Semifinite von Neumann algebras in gauge theory and gravity}

\author{Shadi Ali Ahmad}
\affiliation{Center for Cosmology and Particle Physics, New York University, New York, NY 10003, USA}
\author{Marc S. Klinger}
\affiliation{Illinois Center for Advanced Studies of the Universe \& Department of Physics,\\ 
			University of Illinois, 1110 West Green St., Urbana IL 61801, U.S.A.}
\author{Simon Lin}
\affiliation{Center for Cosmology and Particle Physics, New York University, New York, NY 10003, USA}
\begin{abstract}
    
\end{abstract}
\date{\today}

\begin{abstract}
von Neumann algebras have been playing an increasingly important role in the context of gauge theories and gravity. The crossed product presents a natural method for implementing constraints through the commutation theorem, rendering it a useful tool for constructing gauge invariant algebras. The crossed product of a Type III algebra with its modular automorphism group is semifinite, which means that the crossed product regulates divergences in local quantum field theories. In this letter, we find a sufficient condition for the semifiniteness of the crossed product of a type III algebra with \emph{any} locally compact group containing the modular automorphism group. Our condition surprisingly implies the centrality of the modular flow in the symmetry group, and we provide evidence for the necessity of this condition. Under these conditions, we construct an associated trace which computes physical expectation values. We comment on the importance of this result and its implications for subregion physics in gauge theory and gravity.

\end{abstract}

\maketitle

\section{Introduction}
In quantum field theory, one can assign an algebra of observables to any open set of spacetime which may be weak completed in a Hilbert space to a von Neumann algebra~\cite{haag_local_1993, Araki:1999ar, fewster2019algebraic}. Since von Neuman algebras admit a classification into factors of different types organized by the properties of their lattice of projections~\cite{murray1936rings,murray1937rings, neumann1940rings, murray1943rings, araki1968classification, connes1973theoreme}, a natural question to ask is: what is the type of algebras of observables appearing in quantum field theory? The answer was given in seminal works~\cite{araki_type_1964, araki_lattice_2004,driessler1977type,Fredenhagen:1984dc,buchholz1987universal, buchholz1995scaling, araki_von_2004}, where it was argued that generically such algebras are Type III$_{1}$. One striking property of such algebras is that they are not semifinite, which introduces a plethora of divergences in the computation of physical observables and entropies~\cite{Sorkin:1984kjy, Bombelli:1986, Srednicki:1993,Susskind:1994,McGuigan:1994, Callan:1994,Holzhey:1994}. While this is not a deficiency of quantum field theory itself, it does hinder our ability to quantitatively answer questions involving quantum information within the standard framework. To rigorously approach certain questions would therefore require the replacement of the usual algebras of observables of quantum field theory with semifinite ones endowed with a trace, whose existence is tied to the regulation of the familiar divergences.


The structure theory of Type III$_{1}$ factors was analyzed early on by Takesaki~\cite{takesaki1973duality}, where it was shown that the crossed product of such an algebra with its modular automorphism group is semifinite, even though the original algebra is not. The modular crossed product then allows us to access the semifinite features of the original algebra. Indeed, this result applied to physical scenarios has recently given us an understanding of generalized entropy from a Lorentzian and canonical picture~\cite{witten_gravity_2022, chandrasekaran_algebra_2022,chandrasekaran_large_2022, Jensen:2023yxy, AliAhmad:2023etg, Klinger:2023tgi, Kudler-Flam:2023hkl}. 

By the commutation theorem and its generalizations~\cite{takesaki1973duality,  Takai1975,1977727,1977CMaPh..52..191L, ImaiTakai1978,Pedersen1979,92a701b6-511a-3206-a6c2-c00a4d6f7655, Raeburn1987,Raeburn1988}, the crossed product algebra respects natural constraints associated with the group used to form it. Recent progress in understanding the relevance of the crossed product algebra in the context of gauge theory and gravity has suggested that, in some cases, one can interpret it as the gauge-invariant algebra of observables of the theory~\cite{Klinger:2023tgi,Klinger:2023auu}. Thus, one has incentive to perform the crossed product of algebras of observables by groups that are not necessarily the modular automorphism group~\cite{witten_gravity_2022, AliAhmad:2023etg, Klinger:2023tgi, Klinger:2023auu}.\footnote{Different perspectives on the crossed product in the context of quantum reference frames appear in Refs.~\cite{Gomez:2023wrq,Fewster:2024pur, DeVuyst:2024pop, AliAhmad:2024wja}.} However, the semifiniteness of crossed products with more general groups has, to our knowledge, not been considered in previous literature with the exception of Refs.~\cite{Kudler-Flam:2023qfl,Fewster:2024pur} which consider direct products between $\mathbb{R}$ and a compact group. Without a semifinite trace on the gauge-invariant algebra of observables, one runs again into familiar divergences from quantum field theory when computing entropic observables.

In this letter, we provide a sufficient condition guaranteeing the semifiniteness of the crossed product of a Type III algebra with a locally compact group.
Our result puts the crossed product on more solid physical ground as it allows for the definition of a semifinite trace in gauge theory and gravity. Our Theorem~\ref{thm:main} leverages the theory of dual weights which allows one to understand the modular structure of the crossed product in terms of that of the original algebra~\cite{Digernes:1975,haagerup1978dualI,haagerup1978dualII, Takesaki2}. We phrase our result in terms of physical conditions, namely the existence of a suitably invariant weight on the original algebra satisfying the Kubo–Martin–Schwinger~(KMS) condition relative to a subgroup of the relevant symmetry~\cite{Kubo:1957,Martin:1959, Haag:1967}. We then construct a trace on the resultant semifinite algebra which reduces to the known expression in the modular case. Our condition has a surprising implication: the modular flow must be a central subgroup of the symmetry group. We investigate this further and prove Theorem~\ref{thm:iff} showcasing the \textit{necessity} of this centrality in some cases.

\section{Preliminaries on Modular Theory}

In this section we provide a brief review of the machinery of modular theory. This discussion is meant to bring the reader up to speed with tools we plan to implement in the main analysis of the work such as faithful, semifinite, normal weights, relative modular operators, and the modular automorphism. Naturally, the discussion is not entirely comprehensive. For a detailed exposition we refer the reader to \cite{Digernes:1975,Takesaki2}, or \cite{witten_notes_2018,Sorce:2023fdx} for more physically motivated presentations.

A $C^*$ algebra is an involutive Banach algebra\footnote{Recall that a Banach algebra is an algebra which is also a complete normed vector space.} satisfying the $C^*$ identity
\beq \label{Cstar condition}
	\norm{a^*a} = \norm{a}^2, \; \forall a \in A. 
\eeq
We refer to the topology induced by $\norm{\cdot}$ as the uniform or norm topology. A weight on $A$ is a map $\varphi: A_+ \rightarrow \mathbb{R}_+ \cup \{\infty\}$ satisfying
\beq
	\varphi(a + \lambda b) = \varphi(a) + \lambda \varphi(b), \; \forall a,b \in A, \; \lambda \in \mathbb{R}_+.
\eeq 
Any weight can be uniquely extended to a linear map $\varphi: A \rightarrow \mathbb{C}$, which we regard as an element of the predual $A_*$.\footnote{Recall that the predual of a Banach space $A$ is a Banach space whose Banach space dual is equal to $A$.} The algebra $A$ is termed unital if it contains an identity element. A state on a unital algebra $A$ is a weight which is moreover normalized in the sense that $\varphi(\mathbb{1}) = 1$. In general, a weight may assign an infinite value to an operator $a \in A$. The domain of the weight is defined
\beq
    \mathfrak{n}_{\varphi} \equiv \{a \in A \; | \; \varphi(a^* a) < \infty\}. 
\eeq

By the Gelfand-Naimark theorem \cite{gelfand1994imbedding} every $C^*$ algebra possesses a Hilbert space representation, meaning a faithful map $\pi: A \rightarrow B(H)$. Given such a representation, we can introduce the commutant which is the set of all bounded operators on $H$ which commute with $\pi(A)$:
\beq \label{Commutant}
	\pi(A)' \equiv \{\mathcal{O} \in B(H) \; | \; [\mathcal{O},\pi(a)] = 0, \; \forall a \in A\}. 
\eeq
If $\pi(A)'' = \pi(A)$, the algebra $A$ is called a von Neumann algebra. Note that the double commutant of any $C^*$ algebra, or more generally any algebra represented on a given Hilbert space, always defines a von Neumann algebra \cite{v1930algebra}. Alternatively, a von Neumann algebra can be defined as a $C^*$ algebra which is closed under the weak operator topology induced by a given Hilbert space representation. For this reason, we can think of a von Neumann algebra as a particular subalgebra, $M < B(H)$, for some Hilbert space $H$ such that $M = M''$.

Given a $C^*$ algebra, $A$, and a weight, $\varphi$, there is a natural procedure by which one may obtain a Hilbert space representation and by extension a von Neumann algebra. This procedure is called the Gelfand–Naimark–Segal~(GNS) construction \cite{segal1947irreducible, kadison1986fundamentals}, which we review in the following.

To begin, we define the kernel of the weight $\varphi$:
\beq \label{kernel of varphi}
	k_{\varphi} \equiv \{a \in A \; | \; \varphi(a^*a) = 0\}. 
\eeq
Let $\eta_{\varphi}: \mathfrak{n}_{\varphi} \rightarrow \mathfrak{n}_{\varphi}/k_{\varphi}$ denote the projection from $\mathfrak{n}_{\varphi}$ into the quotient space defined by equivalence up to elements in \eqref{kernel of varphi}. The space $\mathfrak{n}_{\varphi}/k_{\varphi}$ is a pre-closed inner product space with bilinear form
\beq
	g_{\varphi}: \bigg(\mathfrak{n}_{\varphi}/k_{\varphi}\bigg)^{\times 2} \rightarrow \mathbb{C}, \qquad g_{\varphi}\bigg(\eta_{\varphi}(a),\eta_{\varphi}(b)\bigg) = \varphi(a^*b).
\eeq
Thus, we can complete $\mathfrak{n}_{\varphi}/k_{\varphi}$ to a Hilbert space which we denote by $H_{\varphi}$. The algebra $A$ inherits a representation on $\mathfrak{n}_{\varphi}/k_{\varphi}$ from its natural product structure as
\beq \label{GNS rep}
	\pi_{\varphi}: A \rightarrow \mathfrak{n}_{\varphi}/k_{\varphi},\qquad \pi_{\varphi}(a)\bigg(\eta_{\varphi}(b)\bigg) = \eta_{\varphi}(ab), 
\eeq	 
which can be promoted to a representation on $H_{\varphi}$. The set $(H_{\varphi},\pi_{\varphi},\eta_{\varphi})$ is called the GNS triple induced by the weight $\varphi$. The von Neumann algebra induced by $\varphi$ is given by 
\beq \label{GNS vN}
	M_{\varphi} \equiv \bigg(\pi_{\varphi}(A)\bigg)''. 
\eeq

Now, let us turn our attention to weights on von Neumann algebras. One important topology for von Neumann algebras is the ultraweak topology. For $M= B(H)$, this is defined as the weakest topology such that the trace-class operators (the predual of $M$) remain continuous on $M$. Similarly for a general von Neumann algebra $M$, it is the weakest topology for which the predual is continuous on $M$.

As every von Neumann algebra is also a $C^*$ algebra, the definition of a weight is the same in both cases. Let $\varphi$ be a weight on $M$. The weight $\varphi$ is termed \emph{normal} if it is lower semi-continuous relative to the ultraweak topology. Essentially, this means that for a positive operator $a$ in $M$, the values of $\varphi$ on neighboring elements are lower than $\varphi(a)$. Non-trivially, this means that there is a family of \textit{positive} forms $\phi_{i}$ which are ultraweakly continuous on $M$ such that $\varphi(a) = \sum_{i} \varphi_{i}(a)$. The weight is termed \emph{faithful} if $\varphi(x^*x) = 0 \iff x = 0$. Finally, the weight is termed \emph{semifinite} if $\mathfrak{n}_{\varphi}$ is ultraweakly dense in $M$. Note that because of the decomposition of a normal weight into positive forms, this definition can handle the potential infinite values of a semifinite weight $\varphi$. We shall denote the set of faithful, semifinite, normal weights on $M$ by $P(M)$. 

The GNS representation of \emph{any} faithful, semifinite, normal weight on a von Neumann algebra is standard in the sense defined by Haagerup \cite{haagerup1975standard}. For our purpose, the most significant insights from this observation are twofold. Firstly, the GNS Hilbert spaces of any faithful, semifinite, normal weights are isomorphic, $H_{\varphi} \simeq H_{\psi}, \; \forall \varphi,\psi \in P(M)$. Per this observation, we shall undertake the following analysis within the so-called canonical standard form, $L^2(M)$, which is unitarily equivalent to the GNS Hilbert space of any $\varphi \in P(M)$. We denote the standard representation on $L^2(M)$ by $\pi: M \rightarrow B(L^2(M))$. Secondly, for each state $\varphi \in P(M)$ there exists a vector representative $\xi_{\varphi} \in L^2(M)$ such that $\varphi(x) = g_{L^2(M)}(\xi_{\varphi}, \pi(x) \xi_{\varphi})$. These observations set the stage for the study of weights and states on von Neumann algebras, typically referred to as modular or Tomita-Takesaki theory \cite{takesaki1974tomita,takesaki2006tomita}. 

Given a pair of states $\varphi,\psi \in P(M)$ we define the relative Tomita operator $S_{\varphi \mid \psi}: L^2(M) \rightarrow L^2(M)$ such that
\beq \label{Tomita Operator}
	S_{\varphi \mid \psi}\bigg(\pi(x) \xi_{\psi} \bigg) \equiv \pi(x^*) \xi_{\varphi}. 
\eeq
In seminal work by Tomita (and later formalized by Takesaki) it is demonstrated that $S_{\varphi \mid \psi}$ is closeable and possesses a polar decomposition
\beq
	S_{\varphi \mid \psi} = J \Delta_{\varphi \mid \psi}^{1/2}.
\eeq
The operator $J$ is called the modular conjugation. The first fundamental result of Tomita-Takesaki theory is that $J$ intertwines the algebra $M$ and its commutant: $J \pi(M) J = \pi(M)'$. 

The operator $\Delta_{\varphi \mid \psi}$ is positive, non-singular, and self-adjoint and is called the relative modular operator. We also define by $\Delta_{\varphi} \equiv \Delta_{\varphi \mid \varphi}$ the modular operator of the state $\varphi \in P(M)$. The second fundamental result of Tomita-Takesaki theory is that $\Delta_{\varphi}$ generate an automorphism of the algebra $M$. In particular, to each $\varphi \in P(M)$ we associate the map
\beq
	\sigma^{\varphi}: \mathbb{R} \times M \rightarrow M, \; \sigma^{\varphi}_{t}(x) \equiv \pi^{-1}\bigg(\Delta_{\varphi}^{it}\pi(x)\Delta_{\varphi}^{-it}\bigg). 
\eeq
which is called the modular automorphism of $\varphi$. The modular automorphism possesses a more physically motivated characterization in terms of the KMS condition. In particular, the modular automorphism is the unique automorphism which 
\begin{enumerate}
	\item Preserves the state $\varphi$ as $\varphi \circ \sigma^{\varphi}_t = \varphi, \; \forall t \in \mathbb{R}$,
	\item For each $x,y \in M$ admits a bounded function $F_{x,y}$ on the closed horizontal strip $\overline{D} \subset \mathbb{C}$ which is holomorphic in the interior of $\overline{D}$ and satisfies
	\beq
		F_{x,y}(t) = \varphi(\sigma_t^{\varphi}(x)y), \qquad F_{x,y}(t + i) = \varphi(y \sigma^{\varphi}_t(x)). 
	\eeq
\end{enumerate}
The second of these two conditions characterizes when a quantum state corresponds to the thermal equilibirum of a physical system at inverse temperature $\beta = -1$ in which $\sigma^{\varphi}$ defines a thermal evolution \cite{kubo1957statistical, martin1959theory, haag1967equilibrium}. 

The modular automorphisms of any pair of states $\varphi,\psi \in P(M)$ are inner unitarily equivalent. In particular, there exists a one-parameter family of unitary operators $u^{\varphi \mid \psi}_t \in M$ such that
\beq
	\sigma^{\varphi}_t(x) = u^{\varphi \mid \psi}_t \sigma^{\psi}_t(x) u^{\varphi \mid \psi}_{-t}, \; \forall t \in \mathbb{R}. 
\eeq
The operators $u^{\varphi \mid \psi}_t$ are called the Connes' cocycle derivatives of $\varphi$ and $\psi$. They can be constructed using the relative modular operators as \cite{Connes1973,connes1973theoreme}:
\begin{flalign} \label{Connes Cocycle}
	u^{\varphi \mid \psi}_{t} &\equiv \pi^{-1}\bigg(\Delta_{\varphi \mid \psi}^{it} \Delta_{\psi}^{-it}\bigg) \nonumber \\
	&= \pi^{-1}\bigg(\Delta_{\varphi}^{it} \Delta_{\psi \mid \varphi}^{-it}\bigg), \; t \in \mathbb{R}.
\end{flalign}
Under analytic continuation, the Connes' cocycle derivatives can also be used to intertwine the states $\varphi$ and $\psi$:
\beq
	\varphi(x) = \psi(u^{\varphi \mid \psi}_{-i/2}{}^* \; x \; u^{\varphi \mid \psi}_{-i/2}). 
\eeq

\section{Crossed products with locally compact groups} \label{sec: CP}
In this section, we review crossed products with locally compact groups, following Ref.~~\cite{Takesaki2}. The building block is the notion of a dynamical system.
\begin{definition}
    A dynamical system is a triple $(M,G,\alpha)$ where $M$ is a von Neumann algebra\footnote{See \cite{Sorce:2023fdx} for a recent physical review of some elements of von Neumann algebras.}, $G$ is a locally compact group, and $\alpha : G \to \textup{Aut}(M)$ is an action of the group on the algebra via automorphisms. 
\end{definition}
Examples of such an object are algebras of observables of quantum field theories localized to subregions of spacetime~\cite{haag_local_1993} where $G \simeq \mathbb{R}_{\rm mod}$ is the modular automorphism group and $\alpha = \sigma^{\phi}$ is the modular action associated to some weight $\phi: M \to \mathbb{C}$.

To work in a standard quantum picture, we need to represent a dynamical system on a Hilbert space. 
For any dynamical system $(M, G, \alpha)$ with a given Hilbert space representation $\pi: M \rightarrow B(H)$, there exists a canonical extension of $(\pi,H)$ representing both the algebra $M$ and the group $G$ on an extended Hilbert space $L^2(G,H)\simeq H\otimes L^2(G)$ where the automorphism $\alpha$ is implemented unitarily. 
\begin{definition}
    The canonical covariant representation of a dynamical system $(M,G,\alpha)$ induced from the representation $\pi: M \rightarrow B(H)$ is a triple $(L^{2}(G,H),\pi_{\alpha}, \lambda)$ where $\pi_{\alpha}: M \to B(L^{2}(G,H))$ is a representation of $M$ on $L^{2}(G,H)$, and $\lambda: G \to U(L^{2}(G,H))$ is a unitary representation of $G$ on the space given by
    \begin{align}
    \pi_{\alpha}(x) \xi(g):= \pi \circ \alpha_{g^{-1}}(x) \xi(g), \quad \lambda(h) \xi(g):=  \xi(h^{-1}g) \nn
\end{align}
for any $\xi \in L^{2}(G,H)$, $g \in G$, and $x \in M$. The representation is covariant in the sense that the automorphism $\alpha$ is unitarily implemented: $\pi_{\alpha} \circ \alpha_g = \text{Ad}_{\lambda(g)} \circ \pi_{\alpha}$.
\end{definition}
The two representations combine into a new von Neumann algebra called the crossed product. 
\begin{definition}[Crossed product]
    Let $(L^{2}(G,H), \pi_{\alpha}, \lambda)$ be the canonical covariant representation of the dynamical system $(M, G,\alpha)$ induced from the representation $\pi : M \to B(H)$. The crossed product of $M$ with $G$ relative to the automorphism $\alpha$ is the von Neumann algebra 
    \begin{equation} \label{crossed product def}
    M \rtimes_{\alpha} G := \{ \pi_{\alpha}(M), \lambda(G) \}'',
\end{equation}
where the double-prime denotes weak closure in $L^{2}(G,H)$. 
\end{definition}

One useful way of understanding the crossed product is through the generalization of Takesaki duality which we refer to as the commutation theorem~\cite{takesaki1973duality,  Takai1975,1977727,1977CMaPh..52..191L, ImaiTakai1978,Pedersen1979,92a701b6-511a-3206-a6c2-c00a4d6f7655, Raeburn1987,Raeburn1988}.
\begin{theorem}[Commutation theorem] \label{commutation}
  Let $\tilde{M} = M \otimes B(L^{2}(G))$ where $M$ is part of a dynamical system $(M,G,\alpha)$ and define by $\tilde{\alpha} := \alpha \otimes \textup{Ad} \lambda_{r}$ an extension of the automorphism $\alpha$ on $M$ to $\tilde{M}$ with $\lambda_{r}$, the right regular representation of $G$. Then, the fixed-point subalgebra of $\tilde{M}$ under $\tilde{\alpha}$ is isomorphic to $M \rtimes_{\alpha} G$.
\end{theorem}
The commutation theorem allows us to understand the crossed product as an algebra respecting particular constraints imposed by the group $G$. 
This construction can be applied to general gauge theories~\cite{Klinger:2023tgi,Klinger:2023auu} in which case the crossed product is the gauge-invariant algebra of observables. In \cite{Klinger:2023auu} this result has been directly related to more conventional approaches to constraint quantization such as refined algebraic quantization (RAQ) \cite{Giulini:1998kf,Giulini:1998rk,Marolf:2000iq,Ashtekar:1995zh,Marolf:1996gb}, the Becchi-Rouet-Stora-Tyutin~(BRST) method \cite{Becchi:1974xu,Becchi:1974md,Becchi:1975nq,Tyutin:1975qk,Batalin:1981jr,Henneaux:1992ig,Barnich:2000zw,Fuster:2005eg,Jia:2023tki}, and the Faddeev-Popov approach in path integral quantization \cite{Faddeev:1967fc,Stora:1996ip,Cordes:1994fc}. Briefly, the group generators included into the algebra via Eq.~\eqref{crossed product def} furnish apparently gauge non-invariant degrees of freedom which are nonetheless necessary to achieve gauge invariance in the full theory. In RAQ, these degrees of freedom facilitate the group averaging of operators, in BRST they equip the algebra with a representation of the appropriate group cohomology, and in the Faddeev-Popov approach they are used to construct an insertion which localizes the path integral to the constraint surface. For different perspectives on the crossed product in the context of quantum references frames and constraint quantization, see~\cite{fewster_algebraic_2019, DeVuyst:2024pop, AliAhmad:2024wja}.
\section{Semifiniteness of crossed products}
Here, we review dual weights as necessary to prove our main result. In the following $(M,G,\alpha)$ is a dynamical system and $(L^2(G,H),\pi_{\alpha},\lambda)$ is the canonical covariant representation induced from the representation $\pi: M \rightarrow B(H)$. We let $\omega \in M_*$ be a faithful, semifinite, normal weight on the original algebra $M$. We denote by $\Delta_{g,h}: H \rightarrow H$ the relative modular operator and by $u^{g,h}_t \equiv \Delta_{g,h}^{it}\Delta_{h,h}^{-it}\in M$ the Connes cocycle of the weights $\omega_g$ and $\omega_h$. The modular operator of the untranslated weight $\omega$ is given by $\Delta \equiv \Delta_{e,e}$. 

Dual weights are a natural way to associate the weight $\omega$ with a faithful, semifinite, normal weight $\tilde{\omega} \in (M \rtimes_{\alpha} G)_*$ whose modular data is encoded entirely in the modular data of $\omega$ and its $G$-translates. 
\begin{definition}[Haagerup's algebra]
Given a dynamical system $(M,G,\alpha)$ the algebra $M_G$ consists of continuous compactly supported maps from $G$ to $M$. $M_G$ is a $*$-algebra equipped with a product $\star$ 
\begin{equation}
	\big(\mathfrak{X} \star \mathfrak{Y}\big)(g) \equiv \int_{G} \mu(h) \; \alpha_h\big(\mathfrak{X}(gh)\big) \mathfrak{Y}(h^{-1}), \nonumber
\end{equation}
where $\mu(h)$ is the left-invariant Haar measure on $G$, and involution $\flat$ 
\begin{equation}
 \mathfrak{X}^\flat(g) \equiv \delta(g^{-1}) \alpha_{g^{-1}}\big(\mathfrak{X}(g^{-1})\big)^*.  \nn
 \end{equation}
 Here, $\delta: G \to \mathbb{C}$ is the module function of the group $G$ measuring the failure of right invariance of the left Haar measure $\mu$. Namely, $\mu(gh) = \delta(h) \mu(g)$ for $g,h \in G$.
\end{definition}
The algebra $M_G$ can be interpreted as encoding the orbit data of operators in $M$ under the action $\alpha$. It can also be regarded as providing an alternative construction of the crossed product. The covariant representation $(L^2(G,H),\pi_{\alpha},\lambda)$ induces a *-representation $\rho: M_G \rightarrow B(L^2(G,H))$ given by
\begin{equation}
	\rho(\mathfrak{X}) \equiv \int_{G} \mu(g) \; \lambda(g) \pi_{\alpha}(\mathfrak{X}(g)). 
\end{equation}
The algebra $\rho(M_G) \subset B(L^2(G,H))$ is weakly dense in the crossed product $M \rtimes_{\alpha} G$ and thus $\rho(M_G)'' = M \rtimes_{\alpha} G$.  

Having defined the algebra $M_G$\footnote{We refer the reader to Appendix~\ref{app: hf} for more detail.}, we can now state the definition of the dual weight:
\begin{definition}[Dual weight]
Given a faithful, semifinite, normal weight $\phi \in M_*$ there exists a unique faithful, semifinite, normal weight $\tilde{\phi} \in (M \rtimes_{\alpha} G)_*$ satisfying:
\begin{itemize}
	\item For all\footnote{To be completely rigorous, this should be stated in terms of the product of $M_G$ with the ideal $\mathfrak{n}_\phi \equiv \{x \in M \; | \; \phi(x^*x) < \infty\}$.} $\mathfrak{X} \in M_G$, 
	\begin{equation} \label{Value of Dual weights}
		\tilde{\phi}\bigg(\rho(\mathfrak{X}^\flat \star \mathfrak{X})\bigg) = \phi\bigg((\mathfrak{X}^\flat \star \mathfrak{X})(e)\bigg). 
	\end{equation}
	\item For all $\xi \in L^2(G,H)$ the dual modular operator, namely the modular operator of the weight $\tilde{\phi}$, $\tilde{\Delta}: L^2(G,H) \rightarrow L^2(G,H)$ is given by
\begin{equation} \label{dual modular operator}
    \bigg(\tilde{\Delta}^{it}\xi\bigg)(g) = \delta(g)^{it} \pi(u^{g,e}_t) \Delta^{it} \bigg(\xi(g)\bigg).
\end{equation}
\end{itemize}
\end{definition}
The construction of the dual weight is due originally to Digernes \cite{Digernes:1975} for the separable case\footnote{In particular, Digernes' assumed that the group $G$ and the predual $M_*$ are separable. Recall that a topological space is termed separable if it contains a countable, dense subset.} and was later refined by Haagerup \cite{haagerup1978dualI,haagerup1978dualII} for the general case. 

Before stating our results, we make one final definition.
\begin{definition}[Quasi-invariance]
    Let $(M,G,\alpha)$ be a dynamical system and $\omega$ be a faithful, semifinite, normal weight on $M$. If $\omega \circ \alpha_{g} = \delta(g)^{-1} \omega$, we call $\omega$ quasi $\alpha$-invariant. 
\end{definition}
While the condition $\omega \circ \alpha_{g} = \delta(g)^{-1} \omega$ may seem physically unnatural at first sight, we see that in the typical case of a unimodular group, this reduces to requiring the $\alpha$-invariance of $\omega$. Moreover, it could be interpreted as the suitable notion of $\alpha$-invariance for non-unimodular groups~\cite{Duval1990,Duval1991,Giulini1999, Marolf:2000iq}. The dual modular operator of a quasi $\alpha$-invariant weight $\omega$ can be simplified by observing $u^{g, e}_t = \delta(g)^{-it} \mathbb{1}$, which implies $\tilde{\Delta} = \Delta \otimes \mathbb{1}$. The form of this dual modular operator motivates the following theorem:
\begin{theorem}[Semifiniteness]
\label{thm:main}
      Let $(M,G,\alpha)$ be a dynamical system with $(\mathbb{R},+) < G$ a subgroup. We denote by $\gamma: \mathbb{R} \hookrightarrow G$ the inclusion of this subgroup. If (i) there exists a faithful, semifinite, normal weight, $\omega \in M_*$ that is quasi $\alpha$-invariant and (ii) the weight $\omega$ is KMS with respect to the action of $\mathbb{R} < G$, then $M \rtimes_{\alpha} G$ is semifinite.
\end{theorem}
\begin{proof}
    By the uniqueness of the KMS condition, assumption (ii) implies that the modular automorphism $\sigma_t = \alpha_{\gamma(t)}$.
    
    Assumption (i) implies that the dual modular automorphism $\tilde{\sigma}$ is implemented by $\sigma \otimes \mathbb{1}$. It moreover implies by $\sigma^{\omega_g} = \alpha_{g}^{-1} \circ \sigma^{\omega} \circ \alpha_{g}$ that the modular flow $\sigma$ commutes with the automorphism $\alpha$~\cite{herman1970states}. In this case, $\sigma\otimes \mathbb{1}$ is inner implemented in the crossed product $M \rtimes_{\alpha} G$ as proved in Lemma~\ref{central implies inner} of Appendix~\ref{ap: inner}. Then, $\tilde{\sigma}$ is also inner implemented. 
   
   A von Neumann algebra is semifinite if and only if it admits a semifinite normal weight for which its modular automorphism group is innerly implemented~(c.f.~Theorem 3.14 in Ref.~\cite{Takesaki2}), which proves our result.
\end{proof}
A few remarks are in order:
\begin{enumerate}
     \item An advantage of quasi $\alpha$-invariance is that one does not need to assume that the group von Neumann algebra $\mathcal{L}(G)$ is semifinite to obtain the semifiniteness of $M \rtimes_{\alpha} G$. If one had instead required $G$-invariance of $\omega$, the dual modular automorphism would have included a factor of the module function,
    which then requires innerness of the adjoint action of the multiplication operator induced by the module function. This is implied by the semifiniteness of $\mathcal{L}(G)$ in Appendix~\ref{app: Plancherel}, where it is shown that this operator is precisely the modular operator for the Plancherel weight of the group von Neumann algebra.
    \item  The assumptions of our Theorem~\ref{thm:main} imply that $\mathbb{R}< G$ is a central subgroup. The centrality of $\mathbb{R}$ seems to be a necessary condition in some cases as we record in the following theorem: 
    \begin{theorem}[Centrality]\label{thm:iff}
      Let $(M,G,\alpha)$ be a dynamical system with $M$ a type III algebra. Suppose that $M$ admits a faithful, semifinite, normal weight $\omega$ which is KMS with respect to the action of a subgroup $\mathbb{R} < G$. Moreover, assume that the relative commutant $\pi_{\alpha}(M)^{c} \equiv \pi_{\alpha}(M)' \cap (M \rtimes_{\alpha} G)$ is trivial. Then, the crossed product $M \rtimes_{\alpha} G$ is semifinite if and only if $\omega$ is quasi $\alpha$-invariant and so $\gamma(\mathbb{R})$ is a central subgroup.
    \end{theorem}
    \begin{proof}
        Suppose $M \rtimes_{\alpha} G$ is semifinite, which means that the modular automorphism of any faithful, semifinite, and normal weight is inner implemented~\cite{Takesaki2}. In particular, this must hold for $\tilde{\omega}$, the dual weight of $\omega$. In Appendix~\ref{app: central}, we show that under the assumption that $\pi_{\alpha}(M)^{c} = \mathbb{C}$, $\text{Ad}_{\lambda(\gamma(t))}$ is the unique inner implementer of the dual modular automorphism $\tilde{\sigma}$ in $M \rtimes_{\alpha} G$. Now, compare the action of $\text{Ad}_{\lambda(\gamma(t))}$ and $\tilde{\sigma}_{t}$ on an element $\lambda(g)$. The former simply gives rise to $\lambda(\gamma(t)g\gamma(t)^{-1})$, while the latter is given by Eq.~\eqref{dual on group}. Forcing them to be equal implies that
\begin{equation}
 \lambda(g^{-1}\gamma(t)g\gamma(t)^{-1}) = \delta(g)^{it}\pi_{\alpha}[u_{t}^{g|e}]. \nn
\end{equation}
The operators on either side of the above equation are \textit{distinct} generators of the crossed product and so cannot be equal unless they are proportional to the identity. The only way for the operator on the left to be trivial is if $\gamma(\mathbb{R})$ is central in $G$. Finally, the triviality of the operator on the right implies that $\omega \circ \alpha_{g} = \delta(g)^{-1} \omega$ since the cocycle is proportional to the identity if and only if the two weights are proportional. In other words, the weight $\omega$ is quasi $\alpha$-invariant. 

    The reverse direction is proved by our Theorem~\ref{thm:main}.
    \end{proof} 
    \item  Once it has been established that the crossed product is semifinite under the assumptions of Theorem~\ref{thm:main}, it is manifest that the modular flow of $\tilde{\psi} \in (M \rtimes_{\alpha} G)_*$ is inner implemented for every $\psi \in M_*$. If the weight $\psi$ is not quasi $\alpha$-invariant, its modular flow may be implemented by a non-central subgroup of $G$ since the action of its modular group is related to that of $\omega$ by conjugation with cocycles. This conjugation can twist the embedding of $\mathbb{R}$ in $G$, rendering the subgroup non-central. This should be contrasted with Theorem~\ref{thm:iff} which further assumes the triviality of $\pi_\alpha(M)^c$ and therefore forces the modular flow of any KMS weight relative to $\alpha|_{\mathbb{R}}$ to be central.
\end{enumerate}
\section{Constructing traces}
A semifinite von Neumann algebra admits a tracial weight.
In Theorem~\ref{thm:main} and Appendix~\ref{ap: inner} we show that the action of the dual modular flow is innerly implemented by a unitary operator $A^{it}$. In particular,
\begin{equation}
    \tilde{\sigma}_t = \text{Ad}_{A^{it}}, \quad A^{it} = \mathbb{1} \otimes \lambda\big(\gamma(t)\big),
\end{equation}
The operator $A$ is self-adjoint, and its existence is guaranteed by the spectral theorem. Together with the dual weight $\tilde{\omega}$, it allows us to construct a semifinite trace on the crossed product\footnote{Unless $M\rtimes_\alpha G$ is a factor, there will exist multiple inequivalent traces. If it is a factor, there will exist a unique trace up to rescaling. The question of when a general crossed product $M \rtimes_{\alpha} G$ is a factor is open, but a typical sufficient condition (which is supplemented with additional conditions) relies on the outerness of the action $\alpha$. To see why, recall that the crossed product with an inner action is isomorphic to a tensor product algebra, which will have a center.}.
\begin{lemma} \label{lemma:trace}
Let $(M,G,\alpha)$ be a dynamical system satisfying the requirements in Theorem~\ref{thm:main}. The linear map $\tau:M\rtimes_\alpha G \to \mathbb{C}$ given by
\begin{align}
\label{eq:trace}
    \tau(X) &\equiv \tilde{\omega}(A^{-1}X) 
\end{align}
is a faithful semifinite normal tracial weight on $M\rtimes_\alpha G$. 
\end{lemma}
One can rewrite \eqref{eq:trace} in terms of the original weight $\omega$ by making use of the unnormalized ``position eigenket'' $\ket{g}$\footnote{The use of position eigenkets here may seem unorthodox, but can be made mathematically rigorous by the theory of rigged Hilbert spaces \cite{Roberts1996,Wickramasekara_2002}.}, as in
\begin{equation}
\label{eq:trace2}
    \tau(X) = \int_G \mu(g) \delta(g)^{-1} \omega\big(\braket{e|X|g}\big) \braket{e|\lambda(\gamma(i))|g^{-1}},
\end{equation}
where we denote formally $\lambda(\gamma(i))$ to be the self-adjoint generator of the left translation operator.
We will refer to \eqref{eq:trace2} as the \textit{modular trace} on $M\rtimes_\alpha G$. We present a proof of Lemma \ref{lemma:trace} as well as a derivation of \eqref{eq:trace2} in Appendix~\ref{app: trace}.

We now give some examples where our main theorem applies and comment on the physical significance of these cases.
\begin{example}[Modular Group]
    As our first example, consider $G=(\mathbb{R},+)$, which acts on $M$ as the modular automorphism $\sigma$. $M\rtimes_\sigma \mathbb{R}$ is the familiar modular crossed product.
    $\mathbb{R}$ is unimodular with Haar measure $\mu(q)=dq$. The left translation operator is $\lambda(t)=e^{i\hat{p}t}$ where $\hat{p}\equiv -i\frac{d}{dq}$ is the momentum operator on $\mathbb{R}$. 
    If $M$ is a factor, the modular crossed product admits a unique trace given by \eqref{eq:trace2}:
    \begin{align}
    \begin{split}
        \tau(X) &= \int_\mathbb{R} dp \, e^{-p} \, \omega(\braket{0_q|X|p}) 
        = \omega \big(\braket{0_q|Xe^{-\hat{p}}|0_q} \big) \nn,
    \end{split}
    \end{align}
    where we have written the $\mathbb{R}$ integral in momentum space\footnote{In the physics literature, the modular crossed product is often constructed in the \textit{momentum basis} as opposed to the position basis used here. They are related by a simple swap $p\leftrightarrow q$.}. This reproduces the formula given in \cite{witten_gravity_2022,Jensen:2023yxy,chandrasekaran_algebra_2022,AliAhmad:2023etg}.
    
\end{example}

\begin{example}[Direct Product]
Consider $G=\mathbb{R}\times H$ where $\mathbb{R}$ is the modular automorphism group and $H$ is some locally compact group acting on $M$.
This arises in physical systems such as in AdS-Schwarzchild where $H=SO(3)$ implements rotation in a fixed background \cite{witten_gravity_2022}. Let $\mu(h)$ be the left Haar measure on $H$. The left translation operator is again $\lambda(t)=e^{i\hat{p}t}$ and since $G$ is a direct product, $\hat{p}$ commutes with $H$.
This also allows us to decompose $\ket{g}\in G$ as $\ket{g} = \ket{q,h}$ with $q\in \mathbb{R}$ and $h\in H$. The trace can then be decomposed into a double integral, where the integral over $H$ trivializes and we are left with
\begin{align}
\tau(X)
&= \omega \big(\braket{0,e|Xe^{-\hat{p}}|0,e} \big) \nn .
\end{align}
The above trace reduces back to the previous case when $H$ is trivial.
\end{example}

\begin{example}[Quantum Corner Symmetry] 
    Our theorem also holds for groups which are not direct products with $\mathbb{R}$. Let $G = SL_{2} (\mathbb{R}) \rtimes H_{3}$ where $H_{3}$ is the three-dimensional Heisenberg group. This group arises as a central extension of the Extended Corner Symmetry (ECS) group by $\mathbb{R}$~\cite{Ciambelli:2024qgi}. The ECS is the finite part of the Universal Corner Symmetry group~\cite{ciambelli2021isolated}. These groups are relevant symmetries of quantum gravity with isolated corners. Working semiclassically, we may take $M$ to be the algebra of observables of matter and propagating gravitons in a subregion bounded by the corner. Then, supposing there is a weight $\omega$ that is invariant under $G$ and whose modular flow is implemented by the $\mathbb{R}$ central subgroup of $G$, the crossed product $M \rtimes_{\alpha} G$ is semifinite.
\end{example}
\section{Discussion}
A natural generalization of Takesaki's result~\cite{takesaki1973duality} on the semifiniteness of the modular crossed product of a Type III von Neumann algebra would be that the crossed product of such an algebra with \textit{any} locally compact group containing the modular automorphism group is semifinite. In this letter, we have shown that such a statement is not necessarily true, revealing a deeper interplay between modular theory, crossed products, and physical algebras of observables in gauge theory and gravity. Our conclusion is that this is the case whenever $M$ admits a faithful, semifinite, and normal weight that is invariant under the group action and KMS with respect to the action of a subgroup. A surprising implication of this invariance is the centrality of the modular automorphism group in $G$. An almost converse statement of this result is obtained, namely that if we assume the modular automorphism of some weight is implemented in $G$, then the crossed product is semifinite \textit{if and only if} the modular automorphism group is central provided the relative commutant of $M$ is trivial.

Two immediate questions arise from this analysis:
\begin{enumerate}
    \item Is the centrality of the modular automorphism group in $G$ necessary for the semifiniteness of $N_{\alpha} := M \rtimes_{\alpha} G$ when $M$ is Type III? We aim to answer this question in future work by capitalizing on the following observation: take the crossed product of $N_{\alpha}$ itself with its modular automorphism group, $N_{\alpha} \rtimes_{\tilde{\sigma}} \mathbb{R}$. The resultant algebra tensor factorizes as $N_{\alpha} \otimes \mathcal{L}(\mathbb{R})$ if and only if $\tilde{\sigma}$ is inner, and thus $N_{\alpha}$ is semifinite.
    \item How do we deal with the semifiniteness of crossed product algebras where the relevant symmetry group does not contain a $\mathbb{R}$ subgroup, or worse yet, has no center at all? Such a group arises in near horizon physics and JT gravity, where the symmetry is $SL_{2}(\mathbb{R})$ and has a discrete $\mathbb{Z}_{2}$ center~\cite{Lin:2019qwu, Ouseph:2023juq}. One possible resolution is to embed $G$ into a group $G'$ containing $\mathbb{R}$ as a central subgroup, however this may not be physically motivated. Another resolution is that maybe in those cases, the relevant physical algebra of observables is \textit{not} the crossed product, but another subalgebra of the tensor product algebra $M \otimes B(L^{2}(G))$ invariant under a different extended automorphism. 
\end{enumerate}
In conclusion, our work has significantly extended the groups for which the crossed product $M \rtimes_{\alpha} G$ is semifinite when $M$ is type III. Moreover, our analysis has shown that the semifiniteness of crossed products with general locally compact groups is a subtle issue that is worth understanding to truly unravel the importance of crossed product algebras in physics.

\begin{acknowledgments}
We thank Ahmed Almheiri, Thomas Faulkner, Elliot Gesteau, Rob Leigh, and Antony Speranza for useful discussions.
\end{acknowledgments}

\appendix
\onecolumngrid
\section{Haagerup form of the crossed product} \label{app: hf}
Let $(M,G,\alpha)$ be a dynamical system, along with $\pi: M \rightarrow B(H)$ a faithful representation of $M$ on a Hilbert space $H$. We denote by $H_G = L^2(G,H)$ the canonical covariant representation space. The crossed product $M \rtimes_{\alpha} G$ can be obtained as a weak closure of Haagerup's algebra
\beq
	M_G \equiv \{\mathfrak{X}: G \rightarrow M \; | \; \text{continuous, compactly supported}\}. 
\eeq
Let $\rho: M_G \rightarrow B(H_G)$ be the representation induced by the canonical covariant representation $(H_G,\pi_{\alpha},\lambda)$:
\beq \label{Rep of M_G on H_G}
	\rho(\mathfrak{X}) \equiv \int_{G} \mu(g) \; \lambda(g) \pi_{\alpha}(\mathfrak{X}(g)). 
\eeq
Then, $M \rtimes_{\alpha} G = \rho(M_G)''$. In other words, $M_G$ is dense in the crossed product and thus we may establish results for $M \rtimes_{\alpha} G$ by first establishing it for $M_G$ and then carrying into the rest of the algebra by closure. 

The set $M_G$ is made an involutive algebra by endowing it with a product $\star$ and an involution $\flat$
which generalize the group algebra, which is the weak closure of the left regular representation of the group:
\begin{align}\label{prod and inv}
	\big(\mathfrak{X} \star \mathfrak{Y}\big)(g) &\equiv \int_{G} \mu(h) \; \alpha_h\big(\mathfrak{X}(gh)\big) \mathfrak{Y}(h^{-1}), \nn \\ 
 \mathfrak{X}^\flat(g) &\equiv \delta(g^{-1}) \alpha_{g^{-1}}\big(\mathfrak{X}(g^{-1})\big)^*. 
\end{align}
Given \eqref{prod and inv} we can define on $M_G$ an $M$-valued bilinear $\gamma: M_G \times M_G \rightarrow M$:
\beq \label{M valued bilinear}
	\gamma(\mathfrak{X},\mathfrak{Y}) \equiv \big(\mathfrak{X}^\flat \star\mathfrak{Y}\big)(e). 
\eeq
Plugging \eqref{prod and inv} back into \eqref{M valued bilinear} one can show that
\beq
	\big(\mathfrak{X}^\flat \star \mathfrak{Y}\big)(e) = \int_{G} \mu(g) \; \mathfrak{X}(g)^* \mathfrak{Y}(g). 
\eeq
Haagerup's operator valued weight $T: \rho(M_G) \rightarrow M$ is defined as
\beq \label{Haagerup OVW}
	 T\big(\rho(\mathfrak{X}^\flat \star \mathfrak{X})\big) \equiv \gamma(\mathfrak{X},\mathfrak{X}) = \int_{G} \mu(g) \; \mathfrak{X}(g)^* \mathfrak{X}(g).  
\eeq
As mentioned, this defines an operator valued weight on all of $M \rtimes_{\alpha} G$ since $M_G$ is a dense subset. Eqn. \eqref{Haagerup OVW} can be interpreted as the $M$-valued generalization of the Plancherel weight. Given a faithful, semifinite and normal weight $\omega\in M_*$, the dual weight $\tilde{\omega}\in (M\rtimes_\alpha G)_*$ can be simply expressed in terms of $T$ as
\begin{equation}
\label{eq:dual weight-HOVW}
    \tilde{\omega} = \omega \circ T .
\end{equation}
\section{Modular automorphism of the Plancherel Weight} \label{app: Plancherel}
Consider the dynamical system $(\mathbb{C},G,\beta)$ where $\beta$ is a trivial automorphism. Haagerup's algebra $\mathbb{C}_G$ is isomorphic to $C_c(G)$ -- the set of continuous, compactly supported complex valued functions on $G$. This algebra acts on the Hilbert space $L^2(G)$, where the representation $\rho$ given in \eqref{Rep of M_G on H_G} is the action by the group convolutional product. Thus, it is manifest that the crossed product in this case is isomorphic to the group von Neumann algebra: $\mathbb{C} \rtimes_{\beta} G = \rho(\mathbb{C}_G)'' \simeq \mathcal{L}(G)$.

A weight, $\phi$, on $\mathbb{C}$ is identified with a complex number
\begin{equation}
    \phi(z) = \phi^* z, \; \phi \in \mathbb{C}.
\end{equation}
All such weights are tracial and thus have trivial modular automorphism, $\sigma^{\phi}_t(z) = z$ for all $z \in \mathbb{C}$. Eqn. \eqref{Value of Dual weights} identifies dual weights simply as multiples of the Plancherel weight:
\begin{equation}
    \tilde{\phi}\big(\rho(\mathfrak{X}^\flat \star\mathfrak{X})\big) = \phi^* \gamma(\mathfrak{X},\mathfrak{X}) = \phi^* \omega_{G}(\mathfrak{X}^\flat \star \mathfrak{X}). 
\end{equation}
Since the action $\beta$ is trivial the cocycles $u^{g,e}_t = \mathbb{1}$. Thus, we conclude from \eqref{dual modular operator} that the modular operator associated with the Plancherel weight is simply $\bigg(\Delta_{\omega_{G}}^{it} f\bigg)(g) = \delta(g)^{it} f(g)$.
\section{A trace on the crossed product} 
\label{app: trace}
In this appendix we work under the assumptions of Theorem \ref{thm:main}. In particular, $(M,G,\alpha)$ is a dynamical system, $\gamma: \mathbb{R} \hookrightarrow G$ is an inclusion, and $\omega \in M_*$ is a faithful, semifinite, normal weight on $M$ which is quasi $\alpha$-invariant and KMS with respect to $\alpha \circ \gamma: \mathbb{R} \rightarrow \text{Aut}(M)$. The map 
\beq \label{Trace}
	\tau: M \rtimes_{\alpha} G \rightarrow \mathbb{C}, \qquad \tau(X) \equiv \tilde{\omega}\big(A^{-1} X \big),
\eeq
will automatically be a faithful semifinite normal tracial weight on $M \rtimes_{\alpha} G$ provided $A \in M \rtimes_{\alpha} G$ implements the modular automorphism of the weight $\tilde{\omega}$ as $\tilde{\sigma}_t = \text{Ad}_{A^{it}}$. The fact that \eqref{Trace} is a faithful, semifinite and normal follows immediately from the dual weight theorem and the invertiblility of $A$. The fact that it is tracial follows from the KMS condition 
\begin{flalign} \label{Proof of Trace}
\begin{split}
	\tau(XY) &= \tilde{\omega}\big(A^{-1} XY \big) \\
	&= \tilde{\omega}\big(A^{-1} X A A^{-1} Y \big) \\
	&= \tilde{\omega}\big(\sigma_{i}(X) A^{-1} Y \big) \\
	&= \tilde{\omega}\big(A^{-1}Y X\big) = \tau(YX). 
\end{split}
\end{flalign}
To construct the trace more explicitly, note that the dual weight can be written with Haagerup's operator valued weight \eqref{eq:dual weight-HOVW}. Applying \eqref{prod and inv} to the result we have
\begin{align}
\begin{split}
  \tau(X) &= \tilde{\omega} \left( A^{-1}X \right) \\
  &= \omega\bigg(\left(\rho^{-1}(A^{-1})\star \rho^{-1}(X)\right)(e)\bigg) \\
  &= \int_G \mu(g)\, \omega\bigg( \alpha_g\big(\rho^{-1}(A^{-1})(g)\big)\cdot(\rho^{-1}(X)(g^{-1})) \bigg).
\end{split}
\end{align}
We can explicitly write down the inverse map $\rho^{-1}$ by making use of the position eigenket $\ket{g}$\footnote{Heuristically, $\ket{g}$ should be thought of having a ``wave function'' with support only at $g$.} with a formal inner product $\braket{h|g}=\delta(h^{-1}g)$ and resolution of identity $\mathbb{1}=\int_G \mu(g) \ket{g}\bra{g}$. For $\mathfrak{X} \in M_G$ we compute
\begin{align}
\begin{split}
\Braket{e|\rho(\mathfrak{{X}})|g} &= \int \mu(h) \Braket{e|\lambda(h)\pi_\alpha(\mathfrak{X}(h))|g} \\
&= \int \mu(h) \braket{h^{-1}|g}\alpha_{g^{-1}}(\mathfrak{X}(h)) \\
&= \alpha_{g^{-1}}(\mathfrak{X}(g^{-1})),
\end{split}
\end{align}
where we have used the decomposition $\pi_\alpha(X) = \int \mu(g) \; \alpha_{g^{-1}}(X) \otimes \ket{g}\bra{g}$.
Thus, we obtain the relation $\mathfrak{X}(g) = \alpha_{g^{-1}}\left(\Braket{e|\rho(\mathfrak{X})|g^{-1}}\right)$ and we can write the trace as
\begin{align}
\begin{split}
\tau(X) &= \int \mu(g) \,\omega \bigg( \braket{e|A^{-1}|g^{-1}} \alpha_{g}(\braket{e|X|g}) \bigg) \\
&=\int \mu(g) \, \delta(g)^{-1}\omega \big( \braket{e|X|g} \big) \, \braket{e|\lambda(\gamma(i))|g^{-1}}.
\end{split}
\end{align}
To move from the first to the second line we have used the quasi invariance condition of $\omega$ and the relation $A=\mathbb{1}\otimes\gamma(\lambda(i))$.

\section{Dual modular action on generators of the crossed product} \label{ap: inner}
In this appendix, we include the computation of the dual modular automorphism group of the dual weight $\tilde{\omega}$ on the generators of the crossed product $M \rtimes_{\alpha} G$. 
\begin{enumerate}
 \item The algebra representation:
    \begin{align}
          [\tilde{\sigma}_{t}(\pi_{\alpha}(x)) \xi ](s) &= [\tilde{\Delta}^{it} \pi_{\alpha}(x) \tilde{\Delta}^{-it} \xi](s),\nn\\
         &= \alpha_{ts^{-1}}(x) \xi(s),\nn\\
         &= \pi_{\alpha}\circ \alpha_{t} (x) \xi(s),\nn\\
         &= \pi_{\alpha}\circ \sigma_{t} (x) \xi(s) . \implies \tilde{\sigma}_{t}(\pi_{\alpha}(x)) = \pi_{\alpha}\circ \sigma_{t}(x).
    \end{align}

    \item The group representation: 
    \begin{align}
        [\tilde{\sigma}_{t}(\lambda(g)) \xi ](h) &= [\tilde{\Delta}^{it} \lambda(g) \tilde{\Delta}^{-it} \xi](h),\nn\\
        &= \delta(h)^{it}\Delta_{h,e}^{it}[ \lambda(g) \tilde{\Delta}^{-it} \xi](h), \nn \\
       &= \delta(g)^{it}\Delta_{h,e}^{it}[ \tilde{\Delta}^{-it} \xi](g^{-1}h), \nn \\
       &= \delta(g)^{it}\Delta_{h,e}^{it}\Delta_{g^{-1}h,e}^{-it}\xi(g^{-1}h), \nn \\
       &= \delta(g)^{it}\pi\left( u_{t}^{h,g^{-1}h} \right)\xi(g^{-1}h), \nn \\
       &= \delta(g)^{it}\pi \circ \alpha_{h}^{-1}\left( u_{t}^{e,g^{-1}} \right)\xi(g^{-1}h), \nn \\
        &= [\delta(g)^{it}\pi_{\alpha}\left( u_{t}^{e,g^{-1}} \right)\lambda(g)\xi](h), \nn \\
                 &= [\delta(g)^{it}\lambda(g) \lambda(g)^{-1}\pi_{\alpha}\left( u_{t}^{e,g^{-1}} \right)\lambda(g)\xi](h), \nn \\
                   &= [\delta(g)^{it}\lambda(g) \pi_{\alpha} \circ \alpha_{g}^{-1}\left( u_{t}^{e,g^{-1}} \right)\xi](h), \nn \\
                    &= [\delta(g)^{it}\lambda(g) \pi_{\alpha}\left( u_{t}^{g,e} \right)\xi](h),  \implies \tilde{\sigma}_{t}(\lambda(g)) = \delta(g)^{it}\lambda(g) \pi_{\alpha}\left( u_{t}^{g,e} \right). \label{dual on group}
    \end{align}
\end{enumerate}
We now use the above to prove a lemma useful in the proof of Theorem~\ref{thm:main}.
\begin{lemma} \label{central implies inner}
    Suppose the modular automorphism of a weight $\omega$, $\sigma^{\omega}$, commutes with the group action $\alpha: G \to Aut(M)$ on a von Neumann algebra, where $\alpha|_{\mathbb{R}} = \sigma$. Then, $\sigma_{t} = \alpha_{\gamma(t)}$ is inner implemented in $M \rtimes_{\alpha} G$.
\end{lemma}
\begin{proof}
    Beginning from the commutation theorem, we have
    \begin{equation}
        \alpha_{g} \otimes \text{Ad}_{ \lambda_{r}(g)} (X) = X
    \end{equation}
    for any element $X \in M \rtimes_{\alpha} G$. Acting with $\alpha_{g}^{-1} \otimes I$ from the left gives
    \begin{equation}
         I \otimes \text{Ad}_{\lambda_{r}(g)} = \alpha_{g^{-1}} \otimes I.
    \end{equation}
    Consider the action of $\sigma_{t} = \alpha_{\gamma(t)}$ on the two generators of the crossed product
    \begin{align}
    (\sigma_{t} \otimes I) \left[ \lambda(g)\right] &= (I \otimes \text{Ad}_{ \lambda_{r}(\gamma(-t))})(\lambda(g)) = \lambda(g),
\end{align}
where we have used that the left and right regular representations commute, and
\begin{align}
   \left[ (\sigma_{t} \otimes I)  \pi_{\alpha}(x) \xi\right](g)&= (I \otimes \text{Ad}_{ \lambda_{r}(\gamma(-t))})( \pi_{\alpha}(x) \xi)(g) , \nn \\
   &= (I \otimes \text{Ad}_{\lambda(\gamma(t))})(\pi_{\alpha}(x) \xi)(g),\nn \\
   &= (I \otimes \text{Ad}_{\lambda(\gamma(t))})(\pi_{\alpha}(x) \xi)(g),\nn \\
   &= (\pi_{\alpha} \circ \sigma_{t} (x) \xi)(g),\nn
\end{align}
where we have used that $\lambda_{r}(\gamma(t)) = \lambda(\gamma(-t))$ under the assumption that $\mathbb{R} < G$ is central. This may be verified easily on an element $\xi \in L^{2}(G,H)$. Note that we can write the above two results as
\begin{equation}
    (\sigma_{t} \otimes I) \left[ \lambda(g)\right]  = \lambda(\gamma(t)g\gamma(-t)), \quad       (\sigma_{t} \otimes I)  \pi_{\alpha}(x) = \lambda(\gamma(t)) \pi_{\alpha}(x) \lambda(\gamma(-t)).
\end{equation}
We see that $\text{Ad}_{\lambda(\gamma(t))}$ and $\sigma$ agree on the two generators of the crossed product $M \rtimes_{\alpha} G$ and so are identified on the entire algebra. The former is manifestly inner, and so we are done.
\end{proof}
\section{Inner implementer of the dual automorphisms } \label{app: central}
In this appendix, we show that when $\pi_{\alpha}(M)^{c} = \mathbb{C}I$ and $\omega$ is KMS with respect to $\alpha \circ \gamma: \mathbb{R} \rightarrow \text{Aut}(M)$, then $\lambda(\gamma(t))$ is the unique inner implementer of the dual automorphism of $\omega$. Indeed, suppose that there exists some a family of unitary operators $Q_{t} \in M \rtimes_{\alpha} G$ such that 
\begin{equation}
    \tilde{\sigma} = \text{Ad}_{Q_{t}}.
\end{equation}
Then, we may constrain the properties of $Q_{t}$ by looking at the product of two generators
\begin{align}
    \tilde{\sigma}(\pi_{\alpha}(x)\lambda(g)) &= Q_{t}\pi_{\alpha}(x)\lambda(g) Q_{-t}, \nn \\
    &= Q_{t}\pi_{\alpha}(x)Q_{-t} Q_{t}\lambda(g) Q_{-t}, \nn \\
    &= \lambda(\gamma(t))\pi_{\alpha}(x)\lambda(\gamma(t))^{-1}  Q_{t}\lambda(g) Q_{-t}.
\end{align}
This is equivalent to
\begin{equation}
    Q_{-t}\lambda(\gamma(t)) \pi_{\alpha}(x) \lambda(\gamma(t))^{-1} Q_{t} = \pi_{\alpha}(x) \implies X_{t} := Q_{-t}\lambda(\gamma(t)) \in \pi_{\alpha}(M)' \cap M \rtimes_{\alpha}G.
\end{equation}
The newly defined family of unitary operators $X_{t}$ must then lie in the relative commutant of $\pi_{\alpha}(M)$ for $Q_{t}$ to innerly generate the modular flow.
Under the assumption, $X_{t}$ must be proportional to the identity implying that $Q_{t} = \lambda(\gamma(t))$ up to a scalar. In other words, in this case, $\text{Ad}_{\lambda}$ is the unique (up to scaling) inner implementer of the dual modular automorphism.

\end{document}